\documentclass{article}

\usepackage{amsthm,amssymb,amsmath,proof,graphicx,authblk}
\newtheorem{theorem}{Theorem}
\newtheorem{lemma}{Lemma}
\newtheorem{cor}{Corollary}
\theoremstyle{definition}
\newtheorem{definition}{Definition}

\newcommand{\oSigma}{\overline{\Sigma}}
\newcommand{\uSigma}{\underline{\Sigma}}
\newcommand{\ua}{\underline{a}}
\newcommand{\ub}{\underline{b}}
\newcommand{\ob}{\bar{b}}
\newcommand{\oa}{\bar{a}}
\newcommand{\KK}{\mathbf{K}}
\newcommand{\Kc}{\mathcal{K}}
\newcommand{\Bc}{\mathcal{B}}

\newcommand{\Pc}{\mathcal{P}}

\newcommand{\ccdot}{\odot}
\newcommand{\SL}{\mathbin{/}}
\newcommand{\BS}{\mathbin{\backslash}}

\newcommand{\ACT}{\mathbf{ACT}}
\newcommand{\ACTpomega}{\ACT^{+}_{\omega}}
\newcommand{\ACTomega}{\ACT_{\omega}}

\newcommand{\rR}{\mathrm{R}}
\newcommand{\rL}{\mathrm{L}}

\newcommand{\Z}{\mathbf{0}}
\newcommand{\U}{\mathbf{1}}
\newcommand{\MALC}{\mathbf{MALC}}
\newcommand{\MALCp}{\MALC^{+}}

\newcommand{\Var}{\mathrm{Var}}

\newcommand{\yields}{\to}

\newcommand{\plhd}{\mathrel{\mbox{\hspace*{1pt}\raisebox{.75pt}{\scalebox{.4}{$\boldsymbol{+}$}}}\hspace*{-5.7pt}\lhd\hspace*{1pt}}}
\newcommand{\prhd}{\mathrel{\mbox{\hspace*{1pt}\raisebox{.75pt}{\scalebox{.4}{$\boldsymbol{+}$}}}\hspace*{-4pt}\rhd\hspace*{1pt}}}

\newcommand{\preo}{\preccurlyeq}
\newcommand{\oerp}{\succcurlyeq}

\newcommand{\botp}{\bot^{\!\prime}}

\newcommand{\hx}[1]{{\upharpoonright_{#1}}}
\newcommand{\BcL}{\Bc_{L}}
\newcommand{\BcLp}{\Bc_{L}^+}
\newcommand{\BcLph}{\Bc_{L}^+ \hx{h(\bot)}}

\newcommand{\Md}{\mathcal{M}}
\newcommand{\Mf}{\mathfrak{M}}
\newcommand{\Hc}{\mathcal{H}}
\newcommand{\Lf}{\mathfrak{L}}

\newcommand{\ud}{\underline{d}}

\newcommand{\bw}{\alpha}

\begin{document}

\title{On Syntactic Concept Lattice Models for the Lambek Calculus and Infinitary Action Logic}

\author[1,2]{Stepan L. Kuznetsov}
\affil[1]{Steklov Mathematical Institute of Russian Academy of Sciences, Moscow, Russia\\ {\tt sk@mi-ras.ru}}
\affil[2]{HSE University, Moscow, Russia}
\date{}

\maketitle

\begin{abstract}
The linguistic applications of the Lambek calculus suggest its semantics over algebras of formal languages. A straightforward
approach to construct such semantics indeed yields a brilliant completeness theorem (Pentus 1995). However, extending the calculus
with extra operations ruins completeness. In order to mitigate this issue, Wurm (2017) introduced a modification of this semantics, namely,
models over syntactic concept lattices (SCLs). We extend this semantics to the infinitary extension of the Lambek calculus with Kleene
iteration (infinitary action logic), prove strong completeness and some interesting corollaries. We also discuss issues arising with
constants---zero, unit, top---and provide some strengthenings of Wurm's results towards including these constants into the systems involved.

\end{abstract}

\section{Introduction}\label{S:intro}

The \emph{Lambek calculus} was introduced by Joachim Lambek~\cite{Lambek1958} in~1958 for mathematical description of natural language syntax via the framework of Lambek categorial grammars. The idea of categorial grammar itself goes back to earlier work of Ajdukiewicz~\cite{Ajdukiewicz} and Bar-Hillel~\cite{BarHillel}, and Lambek grammars are one of the most well-known classes of categorial grammars.

The three core operations of the Lambek calculus are $\BS$ (left division), $\SL$ (right division), and $\cdot$ (product). Formulae built from variables using these operations are called \emph{syntactic types} or \emph{categories.} Variables are called \emph{primitive types.} The usual primitive types are $s$ (sentence), $np$ (noun phrase), $n$ (common noun group), etc. In a Lambek categorial grammar, each natural language word (lexeme) is associated with a finite number of syntactic types which describe its possible syntactic r\^{o}les in a sentence.

The meaning of product is concatenation. Left division constructs the type $A \BS B$, which consists of all phrases which, if prepended by a phrase of type $A$, yield a phrase of type $B$. Right division, $B \SL A$, is symmetric: here an arbitrary phrase of type $A$ is expected on the right.

A standard example is the sentence {\sl ``John loves Mary,''} where {\sl ``John''} and {\sl ``Mary''} both receive type $np$, and the type for the transitive verb {\sl ``loves''} is $(np \BS s) \SL np$. The fact that {\sl ``John loves Mary''} is a correct sentence is validated by \emph{derivability} of the following \emph{sequent} in the Lambek calculus:
\[
 np, (np \BS s) \SL np, np \yields s.
\]
More interesting examples include dependent clauses: having {\sl ``girl''} of type $n$, {\sl ``the''} of type $np \SL n$, and, most importantly, {\sl ``whom''} of type $(n \BS n) \SL (s \SL np)$, we validate {\sl ``the girl whom John loves''} as $np$ by deriving the following sequent:
\[
n \SL np, n,  (n \BS n) \SL (s \SL np), np, (np \BS s) \SL np \yields s.
\]
Here {\sl ``John loves''} gets type $s \SL np$: if continued by an arbitrary noun phrase, it forms a correct sentence.

In general, a sequent is an expression of the form $A_1, \ldots, A_n \yields B$, where $A_1, \ldots, A_n, B$ are formulae (types) of the Lambek calculus. The Lambek calculus itself is a substructural logical system formulated as an intuitionistic-style sequent calculus; as noticed by Abrusci~\cite{Abrusci1990}, it can be regarded as a non-commutative intuitionistic variant of Girard's~\cite{Girard1987} linear logic (its multiplicative fragment). Inference rules and axioms of the Lambek calculus are listed below in Section~\ref{S:Lambek}.

In the original formulation of the Lambek calculus, left-hand sides of all sequents were required to be non-empty. This requirement, called \emph{Lambek's non-emptiness condition,} is motivated by linguistic validity~\cite[\S~2.5]{MootRetore}. An adjective (like {\sl ``interesting''}), being a modifier for a common noun, receives type $n \SL n$. However, in the absence of Lambek's non-emptiness restriction, the \emph{empty word} also gets this type. This gives the following undesired behaviour. Say, {\sl ``very interesting book''} is a valid common noun group. The empty word being also an `adjective,' however, the grammar would also erroneously validate {\sl ``very book''} as a phrase of type $n$.

From the logical point of view, however, Lambek's non-emptiness restriction looks somehow weird, so versions of the Lambek calculus without this restriction were also considered~\cite{Lambek61}. Throughout this article, we discuss both cases, with and without the restriction.

The linguistical meaning of Lambek formulae, sketched above, suggests the following interpretation of the Lambek calculus on algebras of formal languages, called \emph{L-models.} Formulae are interpreted as formal languages over an alphabet $\Sigma$ (i.e., subsets of $\Sigma^+$ or $\Sigma^*$, depending on whether Lambek's non-emptiness restriction is imposed, that is, whether the empty word is allowed). The three Lambek operations are interpreted as follows ($M, N$ being languages):
\begin{align*}
 & M \cdot N = \{ uv \mid u \in M, v \in N \};\\
 & N \BS M = \{ u \mid (\forall v \in N) \: vu \in M \}; \\
 & M \SL N = \{ u \mid (\forall v \in N) \: uv \in M \}.
\end{align*}

Notice how the properties of divisions depend on whether we allow or disallow the empty word, i.e., on the presence of Lambek's non-emptiness restriction. Even if neither $N$ nor $M$ includes the empty word, it could belong to $N \BS M$ (if the empty word is allowed). Thus, the language $N \BS M$ (and, symmetrically, $M \SL N$) may vary, depending on whether we allow the empty word. In particular, if the empty word is allowed, it always belongs to $N \BS N$, so we have $(N \BS N) \BS M \subseteq M$. This corresponds to the sequent $(p \BS p) \BS q \yields q$, which is derivable only in the absence of Lambek's non-emptiness restriction.

Seminal results by Pentus~\cite{Pentus1995,Pentus1998fmonov} establish completeness of the Lambek calculus w.r.t.\ L-models: a sequent is derivable in the Lambek calculus if{f} it is true in all L-models ($\yields$ being interpreted as $\subseteq$). This result is valid for both cases, with and without the empty word.

Pentus' result, however, is extremely hard to prove. (For the product-free fragment, completeness has a much simpler proof, given by Buszkowski~\cite{Buszkowski1982}.) What is even worse, it fails to propagate to \emph{extensions} of the Lambek calculus with other natural operations on formal languages.

The practical motivation for extending the Lambek calculus is the fact that the generative capacity of Lambek grammars is limited. Namely, Lambek grammars can generate only context-free languages, as shown by Pentus~\cite{Pentus1993LICS}. At the same time, there is some linguistic evidence~\cite{Shieber1985} that natural language can go beyond context-freeness.

Real-life Lambek-based parsers for natural language usually extend the basic calculus with many extra operations. For example, the CatLog system developed by Morrill and his group~\cite{Morrill2019} uses a calculus with more than 40 operations. In this article, we keep on a more theoretical side, and consider extensions of the Lambek calculus with natural language-theoretic operations: union, intersection, and, most importantly, Kleene iteration.

Union and intersection, also called \emph{additive disjunction} and \emph{conjunction} (as opposed to product, which is \emph{multiplicative conjunction}) extend the generative capacity of Lambek grammars by finite intersections of context-free languages~\cite{Kanazawa1991} and even more~\cite{KuznetsovOkhotin2017}. Kleene iteration, in its positive form (i.e., without the empty word), is utilised by Morrill~\cite{Morrill2019} for modelling iterated coordination: in noun phrases like {\sl ``John, Bill, Mary, and Suzy''} the coordinator {\sl ``and''} receives type $(np^+ \BS np) \SL np$. Here ${}^+$ is positive Kleene iteration, which makes {\sl ``and''} accept an arbitrary number of $np$'s on the left. Another usage of iteration, for an alternative account of adjectives, was presented by Buszkowski and Palka~\cite{BuszkowskiPalka2008}.

As noticed above, unfortunately, adding these operations ruins completeness w.r.t.\ L-models. Similar problems arise even with constants for $\varnothing$ (the empty language) and $\{ \varepsilon \}$ (the singleton of the empty word), which seem harmless at first glance. (Concrete counter-examples to completeness are given below in Section~\ref{S:Lambek}.)

Christian Wurm~\cite{Wurm2013,Wurm2017} introduced a modification of L-models, which enjoys much better completeness properties. These models are based on the notion of \emph{syntactic concept lattice} (SCL), which goes back to Clark~\cite{Clark2011}. In what follows, we call these models SCL-models.

SCLs form a specific kind of concept lattices, the notion which is the core of {\em Formal Concept Analysis} (see~\cite{DaveyPriestley91,Wille92,SOKuznetsov1996}). Formal Concept Analysis is widely used in data science, including machine learning, data mining, knowledge representation, etc.
The ideas of SCL are closely related to {\em distributional semantics,} an approach in structural linguistics which studies words and phrases by considering the contexts they could appear in.

The novelty of SCL-models in comparison with L-models (which Wurm calls `canonical' L-models) is the use of a closure operator. Given a fixed language $L$, each language $M$ in an SCL is extended to its closure, denoted by $M^{\rhd\lhd}$. This closure includes all words from $M$ and also all words which appear in the same contexts w.r.t.\ $L$:
\[
v \in M^{\rhd\lhd} \iff \forall (x,y) \in \Sigma^* \times \Sigma^* \: \bigl( (\forall w \in M \: xwy \in L) \Rightarrow xvy \in L \bigr).
\]
This closure operator depends on $L$, as an implicit parameter. To keep the notation concise, this dependence is not made explicit.
Accurate definitions of Wurm's SCL-models are given below. With Lambek's non-emptiness restriction imposed, the words $v,w$ are required to be non-empty; however, words in the \emph{contexts}, $x$ and $y$, may still be
empty.

From the point of view of linguistic meaning of the Lambek calculus, models based on SCLs in fact seem more appropriate than L-models. In Lambek categorial grammars, formulae denote {\em syntactic types,} i.e., grammatic categories assigned to words and phrases. The interpretation of a syntactic category (Lambek formula) should be the set of all words and phrases of the given category: for $np$ it is the set of all noun phrases, for $s$---that of all sentences, for $((np \BS s) \SL np)$---of all transitive verb groups, etc. Linguistic intuition suggests, however, that if a word or phrase can appear exactly in the same contexts as words or phrases of type $A$, then it should also belong to type $A$. For example, {\sl John} is a noun phrase and may appear in contexts like {\sl {\underline{\quad}} loves Mary} or {\sl Mary met {\underline{\quad}} yesterday, } etc. {\sl Pete} may appear in the very same contexts, so it also should belong to the $np$ type. The closure used in SCLs, with $L$ being the set of all correct sentences, exactly augments the language for each type with such possible new words and phrases.

Wurm~\cite{Wurm2017} proved completeness, w.r.t.\ SCL-models, of the Lambek calculus extended with additive conjunction and disjunction (interpreted as intersection and the closure of union). We further extend these results to Kleene / positive iteration, and also discuss issues with constants. Moreover, we prove completeness results in their strong form (i.e., for derivability from sets of hypotheses) and reduce the size of the alphabet to 2 letters (in Wurm's proof, the alphabet was potentially infinite). Finally, we discuss the cases with and without Lambek's restriction. In Wurm's article, the former was just briefly mentioned, without detailed discussion how the
construction gets adapted to this situation. In contrast, we put the case with Lambek's restriction to the first place, as it is (see above) more relevant to linguistic applications of the Lambek calculus.

This article is the journal version of the conference paper~\cite{Kuzn2024WoLLIC}, presented at the 30th Workshop on Logic, Language, Information, and Computation (WoLLIC 2024) in Bern, Switzerland, June 10--13, 2024.
In comparison with~\cite{Kuzn2024WoLLIC}, this article includes the following new material:

\begin{itemize}\setlength{\itemsep}{-2pt}
 \item we first consider systems with Lambek's non-emptiness restriction, as being more appropriate from the linguistic point of view;
 \item we extend most of the results to the case with constant $\top$, meaning the maximal language---this extension requires significant modifications of the constructions;
 \item we also prove a completeness result including the $\bot$ constant for the closure of the empty language, but, as a trade-off, without the unit constant---thus, we show that no constant causes incompleteness on its own.
\end{itemize}

Moreover, given the size constraints of a conference paper, most of the proofs in~\cite{Kuzn2024WoLLIC} were given only as brief sketches. In this article, we present all the proofs in full detail.

\section{Calculi and Models, with Lambek's Non-Emp\-ti\-ness Restriction}\label{S:Lambek}

Let us start by defining the calculi with Lambek's non-emptiness restriction, as the ones which are more suitable from the linguistic point of view (see above). Formulae are constructed from a countable set of variables $\Var$ and constants $\top$ and $\bot$ using five binary operations: $\BS$ (left division), $\SL$ (right division), $\cdot$ (product, or multiplicative conjunction), $\wedge$ (additive conjunction), and $\vee$ (additive disjunction).  Furthermore, we extend the formula language with a unary operation ${}^+$, called positive iteration (written in the postfix form: $A^+$).

Sequents are expressions of the form $A_1, \ldots, A_n \yields B$, where $A_1, \ldots, A_n, B$ are formulae and $n \geq 1$ (non-emptiness condition). Variables will be denoted by letters $p,q,r$ (possibly with subscripts). Capital Latin letters from the beginning of the alphabet stand for formulae. Finally, sequences of formulae are denoted by capital Greek letters. These sequences, being parts of left-hand sides of sequents, could be empty, provided that the whole left-hand side is not.

The intended interpretation of additive operations and constants is lattice-theoretic: $\wedge$ is meet, $\vee$ is join, $\top$ is maximum, $\bot$ is minimum, w.r.t.\ the partial order given by $\yields$. The informal meaning of Lambek operations ($\BS$, $\SL$, $\cdot$) was sketched above in the introduction.

Let us formulate the axioms and inference rules for the \emph{multiplicative-additive Lambek calculus} with Lambek's non-emptiness restriction, which we denote by $\MALCp$.\footnote{The ${}^+$ symbol here reflects the intended language semantics on languages without the empty word, i.e., over $\Sigma^+$. In what follows, however, ${}^+$ will also be used to denote positive iteration; the corresponding extension of $\MALCp$ is called $\ACTpomega$.} These rules are the same as given by Kanazawa~\cite{Kanazawa1991}, who also imposes Lambek's non-emptiness restriction, plus axioms governing constants $\top$ and $\bot$.

\[
 \infer[\mathrm{Id}]{A \yields A}{}
 \qquad
 \infer[\bot\rL]{\Gamma, \bot, \Delta \yields C}{}
 \qquad
 \infer[\top\rR]{\Pi \yields \top}{}
\]
\[
\infer[{\cdot}\rL]{\Gamma, A \cdot B, \Delta \to C}
{\Gamma, A, B, \Delta \to C}
\qquad
\infer[{\cdot}\rR]{\Gamma, \Delta \to A \cdot B}
{\Gamma \to A & \Delta \to B}
\]
\[
\infer[{\BS}\rL]{\Gamma, \Pi, A \BS B, \Delta \to C}
{\Pi \to A & \Gamma, B, \Delta \to C}
\qquad
\infer[{\BS}\rR,\mbox{ where $\Pi$ is non-empty}]{\Pi \to A \BS B}{A, \Pi \to B}
\]
\[
\infer[{\SL}\rL]{\Gamma, B \SL A, \Pi, \Delta \to C}
{\Pi \to A & \Gamma, B, \Delta \to C}
\qquad
\infer[{\SL}\rR,\mbox{ where $\Pi$ is non-empty}]{\Pi \to B \SL A}{\Pi, A \to B}
\]
\[
\infer{\Gamma, A \wedge B, \Delta \to C}
{\Gamma, A, \Delta \to C}
\quad
\infer[{\wedge}\rL]{\Gamma, A \wedge B, \Delta \to C}
{\Gamma, B, \Delta \to C}
\qquad
\infer[{\wedge}\rR]{\Pi \to A \wedge B}
{\Pi \to A & \Pi \to B}
\]
\[
\infer[{\vee}\rL]{\Gamma, A \vee B, \Delta \to C}
{\Gamma, A, \Delta \to C & \Gamma, B, \Delta \to C}
\qquad
\infer{\Pi \to A \vee B}{\Pi \to A}
\quad
\infer[{\vee}\rR]{\Pi \to A \vee B}{\Pi \to B}
\]
\[
\infer[\mathrm{Cut}]
{\Gamma, \Pi, \Delta \to C}
{\Pi \to A & \Gamma, A, \Delta \to C}
\]

Lambek's non-emptiness condition here is implemented in the formulations of ${\BS}\rR$ and ${\SL}\rR$, as these are the only two rules which could potentially violate it.

When talking of derivability without hypotheses (theoremhood), the Cut rule is eliminable. (This is shown by a straightforward extension of Lambek's~\cite{Lambek1958} original argument.) We shall also consider, however, derivability from sets of hypotheses, i.e., finite or infinite sets of sequents added as extra non-logical axioms. (Hypotheses are also required to have non-empty antecedents.) In this case, Cut is, in general, not eliminable, so we keep it as one of inference rules.

As a matter of notation, the $\bot$ constant is sometimes denoted by $\mathbf{0}$ (e.g., as it is done in~\cite{Kuzn2024WoLLIC}). The $\top$ constant may seem redundant, due to the equivalence $\top \leftrightarrow \bot \SL \bot$. As we shall see below, however, $\top$ has better semantic properties, so it is meaningful to keep it as a separate constant with its own axiom, and to consider fragments with $\top$, but without $\bot$.

Next, we extend $\MALCp$ with rules for positive iteration, yielding the system which we denote by $\ACTpomega$ and call \emph{positive infinitary action logic.}
\[
 \infer[{}^+ \rL_\omega]
 {\Gamma, A^+, \Delta \yields C}
 {\bigl( \Gamma, A^n, \Delta \yields C \bigr)_{n=1}^{\infty}}
 \qquad
 \infer[{}^+ \rR_n,\ n \geq 1]
 {\Pi_1, \ldots, \Pi_n \yields A^+}
 {\Pi_1 \yields A & \ldots & \Pi_n \yields A}
\]
This system is a positive variant of infinitary action logic $\ACTomega$, introduced by Palka~\cite{Palka2007}. Cut is also eliminable~\cite{Palka2007}, for derivability without hypotheses. For derivability from hypotheses, $\ACTpomega$ enjoys neither cut elimination, nor compactness.

Derivability of a sequent $\Pi \yields C$ from a set of hypotheses $\Hc$ in the given logic $\Lf$ (e.g., $\MALCp$ or $\ACTpomega$) is denoted as follows:
\( \Hc \vdash_{\Lf} \Pi \yields A \).

The ${}^+ \rL_\omega$ rule is an $\omega$-rule, as it has an infinite sequence of premises. (In the formulation of this rule, $A^n$ means $A, \ldots, A$, $n$ times.) In the presence of an $\omega$-rule, derivations
may be inﬁnite, but are required
to be well-founded (each path from
the goal sequent should reach an axiom in finitely many steps). The usage of $\omega$-rules, or some other infinitary mechanisms, is inevitable, since the set of theorems of $\ACTpomega$ is not recursively enumerable, being $\Pi^0_1$-complete, even in the fragment without $\vee$ and $\wedge$~\cite{Buszkowski2007,Kuzn2017WoLLIC}.
Thus, induction-style axiomatisations for iteration, suggested by Pratt~\cite{Pratt1991} and Kozen~\cite{Kozen1994} (so-called `action logic') are strictly weaker than the $\omega$-rule, and therefore give no hope for completeness.

Being a substructural logic, $\MALCp$ enjoys natural abstract semantics on the appropriate class of residuated algebraic structures (see~\cite{Galatos} for a more general picture).

\begin{definition}
 A \emph{semigroup residuated bounded lattice (SRBL)} is a partially ordered algebraic structure
 $\Kc = (\KK; \preo, \cdot, \BS, \SL, \wedge, \vee, \top, \bot)$, where:
 \begin{enumerate}\setlength{\itemsep}{-2pt}
  \item $(\KK; \preo, \wedge, \vee, \top, \bot)$ is a bounded lattice;
  \item $(\KK; \cdot)$ is a semigroup (i.e., product is associative);
  \item $\BS$ and $\SL$ are residuals of the product w.r.t.\ the $\preo$ partial order:
  \[
   b \preo a \BS c \iff a \cdot b \preo c \iff a \preo c \SL b.
  \]
 \end{enumerate}
\end{definition}

The notion of SRBL is an adaptation of the more well-known notion of residuated (bounded) lattice to the case with Lambek's non-emptiness restriction. Instead of a monoid we now have a semigroup, which corresponds to removing the unit constant, which could have represented an empty left-hand side of a sequent.

Now let us add positive iteration.

\begin{definition}
 A \emph{$*$-continuous positive action lattice ($\omega$PAL)} is a structure
 $\Kc = (\KK; \preo, \cdot, \BS, \SL, \wedge, \vee, \top, \bot, {}^+)$, where:
 \begin{enumerate}\setlength{\itemsep}{-2pt}
  \item $(\KK; \preo, \cdot, \BS, \SL, \wedge, \vee, \top, \bot)$ is an SRBL;
  \item $a^+ = \sup_{\preo} \{ a^n \mid n \geq 1 \}$.
 \end{enumerate}

\end{definition}

It follows from the residuation condition that $(\KK; \cdot, \vee)$ is a semiring.
It is also well-known that $\wedge$, $\vee$, $\cdot$, ${}^+$ in such structures are monotone on both arguments, and $\BS$, $\SL$ are monotone on the numerator and antitone on the denominator.
Finally, it is worth noticing that the standard (Kozen's~\cite{Kozen1981}) definition of $*$-con\-ti\-nu\-ous Kleene algebra uses a stronger version of the $*$-continuity condition:
$b \cdot a^+ \cdot c = \sup \{ b \cdot a^n \cdot c \mid n \geq 1 \}$. In the presence of residuals, however, this follows from the simpler definition.

SRBLs and $\omega$PALs provide standard algebraic semantics for $\MALCp$ and $\ACTpomega$, respectively.

\begin{definition}
An algebraic model $\Md$ is defined via an interpretation function $\alpha$ from formulae to elements of $\KK$, which is defined
arbitrarily on variables, commutes with the operations, and maps constants ($\top$, $\bot$) to the corresponding designated
elements of $\KK$. A sequent $A_1, \ldots, A_n \yields B$ is true in $\Md$ if $\alpha(A_1) \cdot \ldots \cdot \alpha(A_n) \preo
\alpha(B)$ in $\Kc$.
\end{definition}

This definition works both for $\MALCp$ and $\ACTpomega$, as well as all their \emph{elementary fragments,} i.e., systems with restricted
sets of operations. For example, if we keep only $\BS$, $\SL$, and $\cdot$, the corresponding class of algebraic models will be the
class of all residuated semigroups (but not necessarily with a lattice structure).

Now let us recall the standard definitions of entailment and completeness.

\begin{definition}
Let $\Mf$ be a class of models for a given set of operations, and let $\Hc$ be a (possibly infinite) set of sequents. Then $\Hc$
\emph{semantically entails} a sequent $\Pi \yields C$ on the class of models $\Mf$, if for every model $\Md \in \Mf$, in which all
sequents from $\Hc$ are true, $\Pi \yields C$ is also true. Notation: $\Hc \vDash_{\Mf} \Pi \yields A$.
\end{definition}

Let $\Lf$ be one of the logics in question (e.g., $\MALCp$ or $\ACTpomega$).

\begin{definition}
A logic $\Lf$ is \emph{strongly sound} w.r.t.\ the class of models $\Mf$, if (for any set of sequents $\Hc$ and sequent $\Pi \to C$)
$\Hc \vdash_\Lf \Pi \yields C$ implies $\Hc \vDash_{\Mf} \Pi \yields C$.
\end{definition}

Strong soundness of the logics w.r.t.\ corresponding classes of models will hold in all cases considered in this article, and this
is proved by routine checks.

\begin{definition}
A logic $\Lf$ is \emph{strongly complete} w.r.t.\ the class of models $\Mf$, if (for any set of sequents $\Hc$ and sequent $\Pi \to C$)
$\Hc \vDash_{\Mf} \Pi \yields C$ implies $\Hc \vdash_\Lf \Pi \yields C$.
\end{definition}

\begin{definition}
A logic $\Lf$ is \emph{weakly complete} w.r.t.\ the class of models $\Mf$, if any sequent which is true in all models from $\Mf$ is derivable
in $\Lf$ (without hypotheses).
\end{definition}

The standard Lindenbaum--Tarski construction (see, e.g.,~\cite{BlokPigozzi}) establishes strong completeness of $\ACTpomega$ w.r.t.\ the
class of all $\omega$PALs and that of $\MALCp$ w.r.t.\ the class of all SRBLs. More serious completeness issues arise, however, with
\emph{concrete} interesting classes of those general algebraic structures.

L-models, as defined in the introduction, provide a natural class of models for $\MALCp$ and $\ACTpomega$. Namely, these are models over
\emph{algebras of formal languages} without the empty word, and these algebras are $\omega$PALs. In such an algebra, $\KK = \Pc(\Sigma^+)$
for an alphabet $\Sigma$, product and two divisions are defined in the natural language-theoretic way (see Introduction), join and meet
are set-theoretic intersection and union, and positive iteration is defined in the $*$-continuous way:
\[
M^+ = \bigcup_{n=1}^{\infty} M^n.
\]
Finally, $\bot = \varnothing$ and $\top = \Sigma^+$ are, resp., the minimal and maximal languages.

As noticed in the Introduction, however, completeness holds only for some fragments of $\MALCp$. Namely, Buszkowski~\cite{Buszkowski1982}
proved strong completeness for the set of operations $\{ \BS, \SL, \wedge \}$ and Pentus~\cite{Pentus1995} proved weak completeness
for the set $\{ \BS, \SL, \cdot \}$. As noticed by Buszkowski~\cite{Buszkowski2010}, in Pentus' situation strong completeness fails:
we have $p \yields p \cdot p \vDash p \yields q$, but $p \yields p \cdot p \nvdash p \yields q$. Without product, strong completeness
can be extended to a limited use of positive iteration: namely, for formulae in the language of $\BS, \SL, \wedge, {}^+$, where ${}^+$ can be
used only in combinations of the form $A^+ \BS B$ and $B \SL A^+$~\cite{KuznetsovRyzhkova2020}.

In the presence of both meet and join, L-models satisfy the distributivity law $(A \vee B) \wedge C \yields (A \wedge C) \vee (B \wedge C)$, which
is not derivable in $\MALCp$ and therefore even weak completeness fails. Moreover, there are corollaries of the distributivity laws, also not derivable,
in the following fragments: $\{ \BS, \SL, \vee \}$~\cite{KanKuzSce2022IC} and~$\{ \BS, \cdot, \wedge, {}^+ \}$~\cite{KuznetsovRyzhkova2020}.\footnote{For the fragment
$\{\BS,\SL,\wedge,{}^+\}$, without restrictions on the use of ${}^+$, completeness would probably also fail, by an argument similar to the one used in~\cite{KuznetsovRyzhkova2020}. However,
to the best of the author's knowledge, this is still just a conjecture, and proving it remains an open problem.} The constant~$\bot$
also causes problems with completeness: combinations like $\top \SL A$ obey certain structural rules from finite-valued logics, which are invalid in~$\MALCp$
(see~\cite{Kuzn2014}).

The models on syntactic concept lattices (SCLs) are a modification of L-mo\-dels and provide much better completeness properties.

In the presence of Lambek's restriction, we define a modification of SCLs which avoid the empty word, called \emph{positive} SCLs, or PSCLs for short.

We start with an alphabet $\Sigma$ and a fixed language $L \subseteq \Sigma^+$. Define the following two maps: ${}^{\prhd} \colon \Pc(\Sigma^+) \to \Pc(\Sigma^* \times \Sigma^*)$ and ${}^{\plhd} \colon \Pc(\Sigma^* \times \Sigma^*) \to \Pc(\Sigma^+)$, as follows:\footnote{Wurm~\cite{Wurm2017} uses a more cumbersome notation for the positive case: ${}^{+\rhd}$ and ${}^{+\lhd}$.}
\begin{align*}
& M^{\prhd} = \{ (x,y) \in \Sigma^* \times \Sigma^* \mid (\forall w \in M) \: xwy \in L \},\mbox{ for $M \subseteq \Sigma^+$;} \\
& C^{\,\plhd} = \{ v \in \Sigma^+ \mid (\forall (x,y) \in C) \: xvy \in L \},\mbox{ for $C \subseteq \Sigma^* \times \Sigma^*$.}
\end{align*}

It is easy to show that the composition of this two maps, $M \mapsto M^{\prhd\plhd}$, is a closure operator, as well as the dual one: $C \mapsto C^{\,\plhd\prhd}$. 

The closure of a language $M$ (w.r.t.\ $L$) is the following language:
\[
 M^{\prhd\plhd} = \{ v \in \Sigma^+ \mid 
 \forall (x,y) \in \Sigma^* \times \Sigma^* \, \bigl( ( \forall w \in M\ xwy \in L) \Rightarrow xvy \in L \bigr).
\]
Notice that all words in $M^{\prhd\plhd}$ are required to be non-empty. Words $x,y$ appearing in contexts, however, may be empty.
The idea is that we add all words which appear in the same contexts, as words already belonging to $M$. (For a linguistic example, say, if {\sl Mary} belongs to the language of noun phrases, then in its closure we shall also have {\sl Ann}, as this word can occur anywhere where {\sl Mary} can.)

A language $M$ is called closed, if $M = M^{\prhd\plhd}$. Closed languages are in the Galois correspondence (given by the pair of maps ${}^{\prhd}$ and ${}^{\plhd}$) with closed {\em contexts} $C \subseteq \Sigma^* \times \Sigma^*$ (such that $C = C^{\,\plhd\prhd}$).

\begin{definition}
Let $\BcLp$ denote the set of all closed languages:
 $\BcLp = \{ M \subseteq \Sigma^+ \mid M = M^{\prhd\plhd} \}$.
\end{definition}

Upon $\BcLp$, we shall build the structure of \emph{syntactic concept lattice}, which will be an $\omega$PAL (Lemma~\ref{P:correc} below). Notice that this construction is a special case of \emph{residuated frames}, see~\cite{GalatosJipsen2013}.

The partial order on $\BcLp$ is the same subset order ($M_1 \subseteq M_2$) as in the lattice of all formal languages over $\Sigma$. Since for $M_1, M_2 \in \BcLp$ their intersection $M_1 \cap M_2$ is also a
closed language, meet in $\BcLp$ is intersection: $M_1 \sqcap M_2 = M_1 \cap M_2$. Join, however, is not union, as $M_1 \cup M_2$ may not be closed. In order to compute join, one needs to apply closure:
$M_1 \sqcup M_2 = (M_1 \cup M_2)^{\prhd\plhd}$.
The new lattice, $(\BcLp, \cap, \sqcup)$, is not distributive anymore. This removes obstacles for completeness, which are related to distributivity (see above).

Let us also define the replaceability relation $\leqslant_L$: for two words $w_1$ and $w_2$ from $\Sigma^+$,
\[
 w_1 \leqslant_L w_2 \iff
 (\forall (x,y) \in \Sigma^*) \, (xw_2 y \in L \Rightarrow x w_1 y \in L).
\]
For any closed language $M$, if $w_2 \in M$ and $w_1 \leqslant w_2$, then also $w_1 \in M$. Indeed, let $x,y$ be such that $xwy \in L$ for any $w \in M$. In particular, $xw_2 y \in L$, and therefore $xw_1 y \in L$. Hence, $w_1 \in M^{\prhd\plhd} = M$.

Next, we define the remaining operations of an SRBL, and, further, positive iteration. For divisions, we have
\begin{align*}
& M_1 \BS M_2 = \{ u \in \Sigma^+ \mid (\forall v \in M_1) \, vu \in M_2 \},\\
& M_2 \SL M_1 = \{ u \in \Sigma^+ \mid (\forall v \in M_1) \, uv \in M_2 \},
\end{align*}
and it is easy to check that these operations keep the languages closed (namely, if $M_2$ is closed, then $M_1 \BS M_2$ and $M_2 \SL M_1$ are also closed). Thus, $\BS$ and
$\SL$ are defined on $\BcLp$ exactly as on $\Pc(\Sigma^+)$. Product and positive iteration may yield non-closed languages, so closure must be added, as we did for union:
\(
M_1 \circ M_2 = (M_1 \cdot M_2)^{\prhd\plhd}\); \(
M^{\oplus} = (M^+)^{\prhd\plhd}.
\)
The $\top$ constant is again the whole $\Sigma^+$ (this language is already maximal, so closure would not change it). For $\bot$, however, we take not $\varnothing$ itself, but its closure $\varnothing^{\prhd\plhd} = \{ v \in \Sigma^+ \mid (\forall x, y \in \Sigma^*) \, xvy \in L \}$, which may be (depending on $L$) a non-empty language.

\begin{lemma}\label{P:correc}
The algebraic structure $(\BcLp; \subseteq, \circ, \BS, \SL, \cap, \sqcup, {}^\oplus, \Sigma^+, \varnothing^{\prhd\plhd})$ is an $\omega$PAL.
\end{lemma}

The proof is a routine check, so we omit it.

By \emph{PSCL-models} for $\ACTpomega$ (or $\MALCp$) we mean just algebraic models over $\BcLp$ (as an $\omega$PAL or an SRBL).

It would have been natural to conjecture completeness of $\ACTpomega$ (or at least $\MALCp$) w.r.t.\ the class of PSCL-models.
However, as noticed by Wurm~\cite{Wurm2017}, the $\bot$ constant turns out to be problematic. Thus, again following Wurm,
we consider a broader fragment of models, called \emph{PSCL-models with non-standard $\bot$,} for which we shall indeed prove
completeness. Completeness w.r.t.\ the original class of PSCL-models is left as an open problem; for some discussion, see
Section~\ref{S:zero} below.

These modified models also fit in the paradigm of algebraic models on SRBLs or $\omega$PALs, being
models on \emph{upper cones} of PSCLs. For an abstract SRBL / $\omega$PAL $\Kc$, the upper cone of an element $b \in \KK$ is defined as follows:
\[
\KK \hx{b} = \{ a \in \KK \mid b \preo a \}.
\]
In general, this set may be not closed under the operations. There is, however, a simple sufficient condition for this to
be the case.

\begin{definition}
An element $\botp \in \KK$ is called a \emph{local zero} of $\Kc$, if for any $a \oerp \botp$ we have
$a \cdot \botp = \botp \cdot a = \botp$.
\end{definition}

\begin{lemma}
If $\botp$ is a local zero of $\Kc$, then the set $\KK \hx{\botp}$ is closed under all operations of $\Kc$.
\end{lemma}

\begin{proof}
If $a,b \oerp \botp$, then $a \cdot b \oerp a \cdot \botp = \botp$ and $a \cdot \botp = \botp \preo b$, whence $\botp \preo a \BS b$ (the
$\SL$ case is symmetric). Therefore, $\KK \hx{\botp}$ is closed under Lambek operations (product and two divisions). Closure of the upper cone
under meet and join is a general fact of lattice theory. Finally, if $\Kc$ were a $\omega$PAL, we also have positive iteration. If $a \oerp \botp$, then
we have $a^+ \oerp a \oerp \botp$.
\end{proof}

Thus, we may induce operations of $\Kc$ on its subset $\KK \hx{\botp}$, keep $\top$ as $\top$, and take $\botp$ as the new $\bot$, and this
construction gives a new SRBL or $\omega$PAL (depending on what the original $\Kc$ was). Now $\botp$ serves as the new $\bot$.

Let us apply this construction
to PSCL-models of the form $\BcLp$. The old $\bot = \varnothing^{\prhd\plhd}$ is replaced by a language $Z$ such that $Z = Z^{\prhd\plhd}$ and for any
$M \supseteq Z$ we have $(M \cdot Z)^{\prhd\plhd} = (Z \cdot M)^{\prhd\plhd} = Z$.

A trivial example of such $Z$ is $\varnothing^{\prhd\plhd}$ itself. However, depending on $L$, in $\BcLp$ there could exist other local zeroes $Z$. For example, fix two letters
$a,b \in \Sigma$ (e.g., for $\Sigma = \{ a, b, c, d \}$) and let $L = \{ uavb \mid u,v \in \Sigma^* \}$. Then, on one hand, one easily checks that the language  $Z = \{ uav \mid u,v \in \Sigma^* \}$ is a local zero.
On the other hand, in this example $\varnothing^{\prhd\plhd} = \varnothing \ne Z$.

\begin{definition}
A PSCL-model with non-standard $\bot$ for $\ACTpomega$ (resp., $\MALCp$) is a model over an $\omega$PAL (resp., SRBL) of the form $\BcLp \hx{Z}$, where $Z$
is a local zero of $\BcLp$.
\end{definition}

\section{Completeness of $\ACTpomega$ over PSCL-models}

In this section, we are going to prove strong completeness of $\ACTpomega$ over PSCL-models with non-standard $\bot$. The argument basically follows Wurm's one~\cite{Wurm2017}; however,
as noticed above, the construction is changed in order to accommodate $\top$ with its standard interpretation (namely, $\Sigma^+$). In the presence of Lambek's non-emptiness restriction, we
also do not have constant~$\U$, which makes the construction a bit simpler.

We start with the known strong completeness of $\ACTpomega$ over abstract algebraic models (see previous section). So, let us consider an $\omega$PAL $\Kc$ (which, presumably, satisfies
our set of hypotheses and falsifies the goal sequent). We shall transform it into an appropriate $\BcLph$ (the mapping $h$ will be properly defined below).

The set $\KK$ being the domain of $\Kc$, let us introduce two alphabets: $\oSigma = \{ \oa \mid a \in \KK \}$ and $\uSigma = \{ \ub \mid b \in \KK \}$, and let $\Sigma = \oSigma \cup \uSigma$.\footnote{Wurm
uses the same set $\KK$ for $\oSigma$; we perform a renaming for clarity.} For any non-empty word $w = \oa_1 \ldots \oa_n \in \oSigma^+$, let $w^\bullet = a_1 \cdot \ldots \cdot a_n$ (in $\Kc$).
As a matter of notation, we use $w, u, v$ to denote words from $\oSigma^*$; letters $x,y,z,s,t$ denote arbitrary words over $\Sigma^*$. By $|x|_{\uSigma}$ let us denote the number of letters from $\uSigma$ in $x$.

Let the designated language be
\[
 L = \{ w \ub \mid w \in \oSigma^+, w^{\bullet} \preo b \mbox{ in $\Kc$} \} \cup
 \{ x \in \Sigma^+ \mid |x|_{\uSigma} \geq 2 \}.
\]
Thus, we have a PSCL $\BcLp$.

Now let us define, simultaenously, its upper cone and a homomorphism $h$ from $\Kc$ to this upper cone.

The mapping $h$ is defined as follows, given $b \in \KK$:
\[
h(b) = \{ (\varepsilon, \ub ) \}^{\plhd} = \{ x \in \Sigma^+ \mid x \ub \in L \}.
\]
By now, $h$ maps $\Kc$ to $\BcLp$; however, $h(\bot) \ne \varnothing^{\prhd\plhd}$, so it is not a proper homomorphism.
In order to define the upper cone model, take $\botp = h(\bot)$. The following lemmata justify the correctness of the construction.

\begin{lemma}\label{Lm:hb}
 \(
  h(b) = \{ w \in \oSigma^+ \mid w^{\bullet} \preo b \mbox{\ \rm in $\Kc$} \} \cup \{ x \in \Sigma^+ \mid |x|_{\uSigma} > 0 \}.
 \)
\end{lemma}

\begin{proof}
 Let $x \in \Sigma^+$.
 By definition, $x \in h(b)$ if{f} $x \ub \in L$. If $|x|_{\uSigma} > 0$, then $|x\ub|_{\uSigma} \geq 2$, thus $x \ub$ belongs to  $L$ (its second component).
 If $x = w \in \oSigma^+$, then $|w \ub|_{\uSigma} = 1$, so $w \ub \in L$ if{f} $w^\bullet \preo b$ in $\Kc$.
\end{proof}

\begin{lemma}\label{Lm:replace}
Let $v,u \in \oSigma^+$, then $v^{\bullet} \preo u^{\bullet}$ in $\Kc$ if{f} $v \leqslant_L u$.
If $x \in \Sigma^+$ and $|x|_{\uSigma} > 0$, then $x \leqslant_L w$ for any $w \in \oSigma^+$.
\end{lemma}

\begin{proof}
Let $v^{\bullet} \preo u^{\bullet}$, prove that $v \leqslant_L u$. Let $s u t \in L$ for some $s,t \in \Sigma^*$. If $|s u t|_{\uSigma} \geq 2$, then also $|s v t|_{\uSigma} \geq 2$, whence $svt \in L$. Otherwise,
$s u t$ is of the form $w\ub$, i.e., $s = w_1$ and $t = w_2 \ub$, where $w_1, w_2 \in \oSigma^*$. Moreover, $(w_1 u w_2)^\bullet  \preo b$. Since $v^{\bullet} \preo u^{\bullet}$,
we also have $(w_1 v w_2)^\bullet  \preo b$ by monotonicity of product and transitivity. (Notice that here $w_1$ or $w_2$ is allowed to be empty.) Hence, $svt = w_1 v w_2 \ub \in L$.
Now let $v \leqslant_L u$. Let $a = u^{\bullet}$. Take $s = \varnothing$ and $t = \ua$. We have $s u t = u \ua \in L$, and therefore $svt = v\ua \in L$. The latter means $v^{\bullet} \preo a = u^{\bullet}$.

For the second statement, let $s w t \in L$ for some $s,t \in \Sigma^*$. Since $|swt|_{\uSigma} \geq 1$, and $x$ includes at least one more letter from $\uSigma$, while $w$ has none. Therefore,
$|sxt|_{\uSigma} \geq 2$, whence $sxt \in L$.
\end{proof}

\begin{lemma}\label{Lm:hom}
The mapping $h$ preserves all operations of $\Kc$ and maps $\top$ to $\Sigma^+$. Moreover, if $a \preo b$ in $\Kc$, then $h(a) \subseteq h(b)$.
\end{lemma}

\begin{proof} By Lemma~\ref{Lm:hb}, $h(\top)$ includes any word $x$ with $|x|_{\uSigma} > 0$. If $w \in \oSigma^+$, then $w \in h(\top)$ if{f} $w^\bullet \preo \top$. The latter, however, holds
for any $w \in \oSigma^+$. Thus, $h(\top) = \Sigma^+$.

If $a \preo b$ in $\Kc$ and $x \in h(a)$, then either $|x|_{\uSigma} > 0$, and then $x \in h(b)$, or $x = w \in \oSigma^+$ and $w^{\bullet} \preo a$. In the latter case, $w^{\bullet} \preo b$ by transitivity,
so $w \in h(b)$.

Now let us consider the operations of $\Kc$.

\vskip 3pt
\emph{Product:} $h(a \cdot b) = h(a) \ccdot h(b) = \bigl( h(a) \cdot h(b) \bigr)^{\prhd\plhd}$.
Take an arbitrary $x \in h(a \cdot b)$. By Lemma~\ref{Lm:replace}, we have $x \leqslant_L \oa\ob$. Indeed,
if $|x|_{\uSigma} > 0$, this is always the case, and if $x = w \in \oSigma^+$, we use the fact that
$w^{\bullet} \preo a \cdot b = (\oa\ob)^{\bullet}$. Since $\oa\ob \in h(a) \cdot h(b)$ and $x \leqslant_L \oa\ob$, we have
$x \in h(a) \ccdot h(b)$.

For the opposite inclusion, it is sufficient to prove that $h(a) \cdot h(b) \subseteq h(a \cdot b)$, since the latter is closed.
Take a word $x = x_1 x_2$, where $x_1 \in h(a)$ and $x_2 \in h(b)$. If at least one of $x_1$, $x_2$ includes a letter from $\uSigma$, then
so does $x$, and therefore $x \in h(a \cdot b)$ (Lemma~\ref{Lm:hb}). Otherwise $x_1 = w_1$ and $x_2 = w_2$ belong to $\oSigma^+$, and by Lemma~\ref{Lm:hb}
we have $w_1^{\bullet} \preo a$ and $w_2^{\bullet} \preo b$. Hence, $(w_1 w_2)^{\bullet} \preo a \cdot b$, and again by Lemma~\ref{Lm:hb} we have
$x = w_1 w_2 \in h(a \cdot b)$.

\vskip 3pt
\emph{Division:} $h(b \SL a) = h(b) \SL h(a)$. For each $x \in h(b \SL a)$, let us show $x \in h(b) \SL h(a)$. Take any $y \in h(a)$. If either $x$
or $y$ includes a letter from $\uSigma$, so does $xy$, therefore $xy \in h(b)$ by Lemma~\ref{Lm:hb}. Otherwise, $x = w$ and $y = v$ both belong to
$\oSigma^+$, and by Lemma~\ref{Lm:hb} we have $w^{\bullet} \preo b \SL a$ and $v^{\bullet} \preo a$ in $\Kc$. Next, $(wv)^{\bullet} \preo (b \SL a) \cdot a \preo b$,
whence $xy = wv \in h(b)$.

Now take $x \in h(b) \SL h(a)$. If $|x|_{\uSigma} > 0$, then $x$ automatically belongs to $h(b \SL a)$ by Lemma~\ref{Lm:hb}. So let $x = w \in \uSigma^+$. Take $v = \oa \in h(a)$.
We have $w\oa \in h(b)$, whence $w^\bullet \cdot a \preo b$ in $\Kc$, which is equivalent to $w^{\bullet} \preo b \SL a$. The latter means $w \in h(b \SL a)$.

The $\BS$ case is symmetric: though there is an asymmetry in the definition of $L$, the formula for $h(b)$ and the lemmata we use are perfectly symmetric.

\vskip 3pt
\emph{Intersection:} $h(a \wedge b) = h(a) \cap h(b)$. If a word includes a letter from $\uSigma$, it belongs to both sides. Thus, the interesting case is $w \in \oSigma^+$.
Using Lemma~\ref{Lm:hb}, we have $w \in h(a \wedge b)$ if{f} $w^{\bullet} \preo a \wedge b$ if{f} $w^{\bullet} \preo a$ and $w^{\bullet} \preo b$ if{f}
$w \in h(a)$ and $w \in h(b)$.

\vskip 3pt
\emph{Union:} $h(a \vee b) = h(a) \sqcup h(b) = \bigl( h(a) \cup h(b) \bigr)^{\prhd\plhd}$.
Let us show that $\overline{a \vee b}$ (this is one letter) belongs to $\bigl( h(a) \cup h(b) \bigr)^{\prhd\plhd}$. Take a pair $s,t$ such that for any $x \in h(a) \cup h(b)$ the word $sxt$ belongs to $L$. If $st$ includes at least two letters from $\uSigma$, then so does $s\, \overline{a \vee b}\, t$, and therefore the latter belongs to $L$. Otherwise, a letter from $\uSigma$ should be the last letter of $t$ (otherwise we take, say $x = \oa$ and get $sxt \notin L$), i.e., $t = t'\ud$. Now $s$ and $t'$ belong to $\oSigma^*$, and we have $(s \oa t')^{\bullet} \preo d$ and $(s \ob t')^{\bullet} \preo d$. Therefore, $(s \, \overline{a \vee b} \, t')^{\bullet} \preo d$ (here we use the ${\vee}\rL$ rule, which is sound in any SRBL), whence $s \, \overline{a \vee b} \, t' \, \ud  \in L$.

 Now take any $x \in h(a \vee b)$. By Lemma~\ref{Lm:replace}, we
have $x \leqslant_L \overline{a \vee b}$. Therefore, $x \in \bigl( h(a) \cup h(b) \bigr)^{\prhd\plhd}$.

For the opposite inclusion, it is sufficient to prove $h(a) \cup h(b) \subseteq h(a \vee b)$, since the latter is closed. If $x \in h(a)$, then either $|x|_{\uSigma} > 0$ (and then
$x \in h(a \cup b)$, or $x = w \in \oSigma^+$ and $w^{\bullet} \preo a$ in $\Kc$. In the latter case, we also have $w^{\bullet} \preo a \vee b$, whence $w \in h(a \vee b)$. The $h(b)$ case
is symmetric.

\vskip 3pt
\emph{Positive iteration:} $h(a^+) = \bigl( h(a) \bigr)^{\oplus} = \bigl( (h(a))^+ \bigr)^{\prhd\plhd}$. The argument is similar to the one for union. First, we show that $\overline{a^+}$ (again, this is one letter) belongs to $\bigl( (h(a))^+ \bigr)^{\prhd\plhd}$. Next, if $x \in h(a^+)$, then by Lemma~\ref{Lm:replace} we have $x \leqslant_L \overline{a^+}$, whence $x \in \bigl( (h(a))^+ \bigr)^{\prhd\plhd}$. For the opposite inclusion, we take $x = x_1 \ldots x_n \in h(a)^n$ (for some $n \geq 1$). If at least one $x_i$ includes a letter from $\uSigma$, then so does $x$, and it belongs to $h(a^+)$. Otherwise, each $x_i = w_i \in \oSigma^+$, and $w_i^{\bullet} \preo a$ for each $i$. Then we have $(w_1 \ldots w_n)^{\bullet} \preo a^n \preo a^+$, whence $x \in h(a^+)$.
%Let
%$x \in h(a \cdot b)$, let us show that $x \in M = \bigl( h(a) \cdot h(b) \bigr)^{\prhd\plhd}$. If $x \notin \oSigma^+$,
\end{proof}

\begin{lemma}\label{Lm:lz}
$h(\bot)$ is a local zero of $\BcLp$.
\end{lemma}

\begin{proof}
Take any $M \in \BcLp$ such that $h(\bot) \subseteq M$. (Recall that the preorder of $\BcLp$ is the subset relation.) In particular,
for any $a \in \KK$ we have $h(\bot) \subseteq h(a)$, since $\bot \preo a$ in $\Kc$ (by Lemma~\ref{Lm:hom}).
First, since $\bot = \bot \cdot \bot$ in $\Kc$, by Lemma~\ref{Lm:hom} we have $h(\bot) = h(\bot) \ccdot h(\bot)$, and, further by
monotonicity, $h(\bot) \ccdot h(\bot) \subseteq M \ccdot h(\bot)$ and $h(\bot) \ccdot h(\bot) \subseteq h(\bot) \ccdot M$.

Now it remains to prove the opposite inclusions: $M \ccdot h(\bot) \subseteq h(\bot)$ and $h(\bot) \ccdot M \subseteq h(\bot)$.
Since $h(\bot)$ is closed, it is sufficient to prove $M \cdot h(\bot) \subseteq h(\bot)$ (the second one, $h(\bot) \cdot M \subseteq h(\bot)$, is symmetric).
Let $x \in M$ and $y \in h(\bot)$. If at least one of these words includes a letter from $\uSigma$, then so does $xy$, and therefore $xy \in h(\bot)$ by
Lemma~\ref{Lm:hb}. Otherwise, $x = u$, $y = w$, both from $\oSigma^+$. By Lemma~\ref{Lm:hb}, since $w \in h(\bot)$, we have $w^{\bullet} \preo \bot$.
Then $(uw)^{\bullet} \preo u^{\bullet} \cdot \bot = \bot$, and therefore $xy = uw \in h(\bot)$.
\end{proof}

This lemma allows us to correctly define $\BcLph$ and consider a model over it. Let $\beta$ be the interpretation function over the original
$\omega$PAL $\Kc$. Then the new interpretation function, over $\BcLph$, is defined as follows, for any formula $A$:
\(
\alpha(A) = h(\beta(A)).
\)
By Lemma~\ref{Lm:hom} and the fact that $h(\bot) = \botp$, we see that $h$ is a homomorphism of $\Kc$ to $\BcLph$. Therefore,
$\alpha$ is a correctly defined interpretation function.

Now everything is ready to prove strong completeness.

\begin{theorem}\label{Th:PSCL}
$\ACTpomega$ is strongly complete w.r.t.\ PSCL-models with non-standard $\bot$.
\end{theorem}

\begin{proof}
Consider a set of sequents (hypotheses) $\Hc$ and a sequent $\Pi \yields C$. Suppose $\Hc \nvdash \Pi \yields C$.
By algebraic completeness, let $\Md$ be a model over an $\omega$PAL $\Kc$, such that all sequents from $\Hc$ are true in $\Md$, but $\Pi \yields C$ is not.
Using the construction above, we build the corresponding PSCL $\BcLp$ and its upper cone $\BcLph$. The new model $\Md'$ is given by the interpretation function
$\alpha(A) = h(\beta(A))$, where $\beta$ was the original interpretation function. We claim that $\Md'$ makes true exactly the same sequents as $\Md$ does.
In particular, this means that all sequents from $\Hc$ are true in $\Md'$, but $\Pi \yields C$ is not. So, $\Hc \nvDash \Pi \yields C$ on the class of PSCL-models with
non-standard $\bot$, thus we obtain strong completeness by contraposition.

Let a sequent $A_1, \ldots, A_n \yields B$ be true in the original (abstract) model $\Md$, i.e., $\beta(A_1) \cdot \ldots \cdot \beta(A_n) \preo \beta(B)$. Since
$h$ is a homomorphism, we also have $h(\beta(A_1)) \circ \ldots \circ h(\beta(A_n)) \subseteq h(\beta(B))$, i.e., $\alpha(A_1) \cdot \ldots \cdot \alpha(A_n) \subseteq \alpha(B)$.
This means that the sequent in question is true in $\Md'$. For the opposite direction, let $h(\beta(A_1)) \circ \ldots \circ h(\beta(A_n)) \subseteq h(\beta(B))$ and denote
$\beta(A_1)$, \ldots, $\beta(A_n)$, $\beta(B)$ by $a_1, \ldots, a_n, b$ respectively (those are elements of $\KK$). Each $\oa_i$ belongs to $h(a_i)$, therefore
$\oa_1 \ldots \oa_n \in h(a_1) \cdot \ldots \cdot h(a_n) \subseteq h(a_1) \circ \ldots \circ h(a_n) \subseteq h(b)$. By Lemma~\ref{Lm:hb}, $a_1 \cdot \ldots \cdot a_n \preo b$,
which means that $A_1, \ldots, A_n \yields B$ is true in $\Md$.
\end{proof}

\section{Abolishing Lambek's Non-Emptiness Restriction}\label{S:norestr}

The system $\MALC$, without Lambek's non-emptiness restriction (also called `full Lambek calculus,' denoted by $\mathbf{FL}_{\bot}$), is obtained from $\MALCp$ (defined in Section~\ref{S:Lambek}) in the following way.

First, we allow sequents of the form $\yields B$ (i.e., with empty left-hand sides) everywhere, including hypotheses, and remove the ``$\Pi$ is non-empty'' constraint on the rules ${\BS}\rR$ and ${\SL}\rL$. Second, we add the unit constant~$\U$ with the following axiom and rule:
\[
 \infer[\U \rL]
 {\Gamma, \U, \Delta \yields C}
 {\Gamma, \Delta \yields C}
 \qquad
 \infer[\U \rR]
 {\yields \U}{}
\]

It is well-known that this extension is not conservative: e.g., the sequent $(p \BS p) \BS q \yields q$ itself has a non-empty left-hand side, but it is derivable only in $\MALC$, not in $\MALCp$.

The class of models should be adjusted accordingly:
\begin{definition}
 A \emph{residuated bounded lattice} (RBL) is a partially ordered algebraic structure
 $\Kc = (\KK; \preo, \cdot, \BS, \SL, \wedge, \vee, \top, \bot, \U)$, where:
 \begin{enumerate}\setlength{\itemsep}{-2pt}
  \item $(\KK; \preo, \cdot, \BS, \SL, \wedge, \vee, \top, \bot)$ is a semigroup residuated bounded lattice (SRBL);
  \item $\U$ is the unit for product, i.e.,
  $(\KK; \cdot, \U)$ is a monoid.
 \end{enumerate}
\end{definition}

Not every semigroup has a unit, so the class of RBLs is narrower than the class of SRBLs.

Next, we could have added positive iteration ($A^+$) to $\MALC$ without Lambek's restriction, and this would give an alternative formulation of infinitary action logic $\ACTomega$. The conventional way in the presence of the unit, however, is to replace $A^+$ with \emph{Kleene iteration} $A^*$, which is equivalent to $\U \vee A^+$.

The explicit rules for Kleene iteration, which turn $\MALC$ into $\ACTomega$, are as follows:
\[
 \infer[{}^* \rL_\omega]
 {\Gamma, A^*, \Delta \yields C}
 {\bigl( \Gamma, A^n, \Delta \yields C \bigr)_{n=0}^{\infty}} \qquad
 \infer[{}^* \rR_n,\ n \geq 0]
 {\Pi_1, \ldots, \Pi_n \yields A^*}
 {\Pi_1 \yields A & \ldots & \Pi_n \yields A}
\]
(Positive iteration is now expressible as $A^+ = A \cdot A^*$.)

\begin{definition}
 An \emph{infinitary action lattice} ($\omega$AL) is a partially ordered algebraic structure $\Kc = (\KK; \preo, \cdot, \BS, \SL, \wedge, \vee, \top, \bot, \U, {}^*)$, where:
 \begin{enumerate}\setlength{\itemsep}{-2pt}
  \item $(\KK; \preo, \cdot, \BS, \SL, \wedge, \vee, \top, \bot, \U)$ is an RBL;
  \item $a^* = \sup_{\preo} \{ a^n \mid n \geq 0 \}$.
 \end{enumerate}
\end{definition}

The standard argument shows strong completeness of $\MALC$ w.r.t.\ the class of all RBLs and that of $\ACTomega$ w.r.t.\ the class of all $\omega$ALs.

L-models on $\Pc(\Sigma^*)$ (i.e., allowing the empty word) form a natural class of models for $\ACTomega$ and $\MALC$. Again, additive disjunction and constants cause incompleteness, even in the weak sense. A counter-example to strong completeness with product for this case was given in~\cite{AndrekaMikulas1994}:
two hypotheses $p \yields p \cdot p$ and $q \yields p$ semantically entail $q \yields p \cdot q$, but derivability fails.

The more advanced class of SCL-models is defined similarly to PSCL-models. The difference is that now we allow the empty word. Again, $L \subseteq \Sigma^*$ is a fixed language, and the two maps (Galois connection) are defined as follows:
\begin{align*}
& M^{\rhd} = \{ (x,y) \in \Sigma^* \times \Sigma^* \mid (\forall w \in M) \: xwy \in L \},\mbox{ for $M \subseteq \Sigma^*$;} \\
& C^{\lhd} = \{ v \in \Sigma^* \mid (\forall (x,y) \in C) \: xvy \in L \},\mbox{ for $C \subseteq \Sigma^* \times \Sigma^*$.}
\end{align*}
(Actually, ${}^\rhd$ is the same as ${}^{\prhd}$, but ${}^{\lhd}$ is essentially different from ${}^{\plhd}$, as now $v$ may be empty.)

The set of all closed languages (such that $M = M^{\rhd\lhd}$) is denoted by $\BcL$, and it bears the structure of an $\omega$AL:
$(\BcL; \subseteq, \circ, \BS, \SL, \cap, \sqcup, {}^{\circledast}, \Sigma^*, \varnothing^{\rhd\lhd}, \{\varepsilon\}^{\rhd\lhd} )$, where $M^{\circledast} = (M^*)^{\rhd\lhd}$.

Next, we again dodge the problems with $\bot$ by considering upper cones of the form $\BcL \hx{Z}$, where $Z$ is a local zero of $\BcL$ \emph{and} $\varepsilon \in Z$ (this guarantees that the unit $\{ \varepsilon \}^{\rhd\lhd}$ is still in the structure). Models on such structures are called SCL-models with non-standard $\bot$. And again, we prove strong completeness:

\begin{theorem}\label{Th:SCL}
 $\ACTomega$ is strongly complete w.r.t.\ SCL-models with non-standard $\bot$.
\end{theorem}

The construction is very similar to the one used to prove Theorem~\ref{Th:PSCL}. The core is the designated language, which is as follows:
\[
 L = \{ w \ub u \mid w, u \in \oSigma^*, w^\bullet \preo b, u^\bullet \preo \U \mbox{ in $\Kc$} \} \cup \{ x \in \Sigma^* \mid |x|_{\uSigma} \geq 2 \}.
\]
The difference is in the word $u$. This addition goes back to Wurm~\cite{Wurm2017} and is needed for words from $h(\U)$, which has to be $\{ \varepsilon \}^{\rhd\lhd}$, to be able to be appended to any word from $L$.

The mapping $h$ is the same: $h(b) = \{ (\varepsilon, \ub)\}^{\lhd} = \{ x \in \Sigma^* \mid x \ub \in L \}$.

Again, we show that $h(b) = \{ w \in \oSigma^* \mid w^{\bullet} \preo b \} \cup \{ x \in \Sigma^* \mid |x|_{\uSigma} > 0 \}$ (analog of Lemma~\ref{Lm:hb}) and an analog of Lemma~\ref{Lm:replace}. In the latter, replacement may occur not only in $w$, but also in $u$.

The main lemma is now as follows:
\begin{lemma}
 The mapping $h$ preserves all operations of $\Kc$, maps $\top$ to $\Sigma^*$ and $\U$ to $\{\varepsilon\}^{\rhd\lhd}$. Moreover, if $a \preo b$, then $h(a) \subseteq h(b)$.
\end{lemma}

\begin{proof}
 Most of the proof directly copies the argument from the proof of Lemma~\ref{Lm:hom}, so we omit it.

 A slight different occurs for the union. Namely, if the context $s,t$ includes a letter from $\uSigma$, now it is not necessarily the last one.  If this letter is in $t = t' \ud u$, we have $u^{\bullet} \preo \U$, $(s \oa t')^{\bullet} \preo d$ and $(s \ob t')^{\bullet} \preo d$. Therefore, $(s \, \overline{a \vee b} \, t')^{\bullet} \preo d$, whence $s \, \overline{a \vee b}\, t \in L$. If the letter is in $s = w \ud s'$, then, conversely, $w^{\bullet} \preo d$, $(s' \oa t)^{\bullet} \preo \U$ and $(s' \ob t)^{\bullet} \preo \U$, and we continue in the same fashion, with $\U$ instead of $d$. The same change should be made in the argument for positive iteration.

 The new cases are unit and Kleene iteration (instead of positive iteration).

 \vskip 5pt
 \emph{Unit:} $h(\U) = \{ \varepsilon \}^{\rhd\lhd}$. First, since $\varepsilon^{\bullet} = \U$, we have $\varepsilon \in h(\U)$, and by closure $\{ \varepsilon \}^{\rhd\lhd} \subseteq h(\U)$. For the opposite inclusion, take $x \in h(\U)$ and show that for any $s,t$ if $st \in L$ then $sxt \in L$. Notice that $|st|_{\uSigma} \geq 1$, so if $x$ includes a letter from $\uSigma$, then $|sxt|_{\uSigma} \geq 2$, whence $sxt \in L$. The same happens if $|st|_{\uSigma} \geq 2$.

 The interesting case is where $|st|_{\uSigma} = 1$ and $x = v \in \oSigma^*$ with $v^{\bullet} \preo \U$. Then either $s = w_1 \in \oSigma^*$ and $t = w_2 \ub u$ or $s = w \ub u_1$ and $t = u_2 \in \oSigma^*$. In both cases, inserting $v$ does not alter the corresponding condition: $w^{\bullet} \preo b$ or $u^{\bullet} \preo \U$.

\vskip 5pt
\emph{Kleene iteration:} here we use the fact that $A^*$ is equivalent to $\U \vee A^+$ and reduce Kleene iteration to the unit, union, and positive iteration, all considered above.
\end{proof}

Having proved this lemma, we further proceed exactly as in the proof of Theorem~\ref{Th:PSCL}. For sequents with empty left-hand sides, we have the following. The sequent $\yields B$ is true in the original abstract model $\Md$, if $\U \preo \beta(B)$. Applying $h$ yields $h(\U) \preo h(\beta(B))$. Since $h(\U) = \{ \varepsilon \}^{\rhd\lhd}$, this is true if{f} $\varepsilon \in h(\beta(B))$, which is the truth conditon for ${} \yields B$.

\section{Variants and Corollaries of Completeness}

This section contains some discussion on variations and corollaries of the completeness results presented above.

First, let us show how to reduce the size of the alphabet $\Sigma$. In the definition of (P)SCL-models, the alphabet was allowed to be infinite. The model constructions used to prove completeness essentially use this feature: $\Sigma$ includes two copies of each element of $\KK$, the domain of an abstract algebraic model. This domain could be made countable, but, in general, not finite.

For weak completeness, we could have used the finite model property (FMP) and indeed make the domains finite. For strong completeness, as we shall see below, the FMP does not hold. On the other hand, in linguistic practice alphabets are usually finite.

We shall show how to reduce an arbitrary SCL with a countable (or finite) $\Sigma$ to an SCL over a two-letter alphabet $\Sigma_2 = \{ e,f \}$. The reducing function is due to Pentus~\cite{Pentus1995}, and we show that it can be transferred from L-models to SCL-models. Unfortunately, it does not work for constants $\U$ and $\top$, making their interpretation non-standard (for $\bot$, it was non-standard from the start).

We shall present the argument for the system without Lambek's restriction. For the other case, the proof is basically the same.

Let $\Sigma = \{ a_1, a_2, \ldots \}$ be an infinite alphabet and let $g \colon \Sigma^* \to \Sigma_2^*$ be a homomorphism defined (on letters) as follows:
\(
 g(a_i) = e f^i e.
\)
The homomorphism $g$ can also be applied to languages: $g(M) = \{ g(w) \mid w \in M \}$.
Notice that $g$ is injective: every word of the form $ef^{i_1}eef^{i_2}e \ldots ef^{i_n}e$ is the image of exactly one word, namely, $a_{i_1} a_{i_2} \ldots a_{i_n}$.

\begin{lemma}\label{Lm:grhdlhd}
If $M \ni w \ne \varepsilon$ and $M^{\rhd} \ne \varnothing$, then $g(M^{\rhd\lhd}) = \bigl(g(M)\bigr)^{\rhd\lhd}$. The closure operator in $\Bc_{L_0}$ and in $\Bc_{g(L_0)}$, resp., is taken w.r.t.\ $L_0$ and $g(L_0)$.
\end{lemma}

\begin{proof}
Let $v \in g(M^{\rhd\lhd})$, i.e., $v = g(v')$, where $v' \in M^{\rhd\lhd}$. Take a pair $(x,y) \in \Sigma_2^* \times \Sigma_2^*$, such that for any $w \in g(M)$ we have $xwy \in g(L_0)$. There is a non-empty $w \in g(M)$, therefore, $x$ and $y$ are also of the form $g(x')$ and $g(y')$ respectively.  Moreover, for any $w' \in M$ we have $x'w'y' \in L_0$. Therefore, since $v' \in M^{\rhd\lhd}$, we have $x'v'y' \in L_0$, whence $xvy \in g(L_0)$. This yields $v \in \bigl(g(M)\bigr)^{\rhd\lhd}$. The inclusion $g(M^{\rhd\lhd}) \subseteq \bigl(g(M)\bigr)^{\rhd\lhd}$ is established.

Now let $v \in \bigl(g(M)\bigr)^{\rhd\lhd}$. Since $M^\rhd$ is non-empty, there exists a pair $(x',y') \in \Sigma^* \times \Sigma^*$, such that $x'w'y' \in L_0$ for any $w' \in M$. Apply the mapping $g$: $x = g(x')$, $y = g(y')$. For the pair $(x,y)$, we have $xvy \in g(L_0)$. Hence, $v$ is of the form $g(v')$ for some $v' \in \Sigma^*$. Let us show that $v' \in M^{\rhd\lhd}$. Take an arbitrary pair $(x',y') \in M^{\rhd}$ and apply $g$ as above. For any $w' \in M$, we have $x'w'y' \in L_0$, whence for $w = g(w')$ we have $xwy \in g(L_0)$. Since $v \in \bigl(g(M)\bigr)^{\rhd\lhd}$ and $w$ is an arbitrary element of $g(M)$, we get $xvy \in g(L_0)$, whence $x'v'y' \in L_0$. This yields $v' \in M^{\rhd\lhd}$, and therefore $v \in g(M^{\rhd\lhd})$. This establishes the inclusion $g(M^{\rhd\lhd}) \supseteq \bigl(g(M)\bigr)^{\rhd\lhd}$.
\end{proof}
Both non-emptiness conditions here are crucial. In what follows, we shall check them when applying this lemma. (Similar issues were dealt with by Pentus.)

\begin{theorem}\label{Th:twoletter}
The fragment of $\ACTomega$ without constants ($\top$, $\bot$, and $\U$) is strongly complete w.r.t.\ SCL-models over the two-letter alphabet $\Sigma_2$.
\end{theorem}

\begin{proof}
The set of variables is countable, whence so are the RKLs given by the Lindenbaum--Tarski construction. Thus, the alphabet $\Sigma$ used in our SCL-models  is also countable. We shall reduce $\Sigma$ to $\Sigma_2$ using $g$.

Let us inspect the proof of Theorem~\ref{Th:SCL}. In the construction of the SCL-model, each formula $A$ is interpreted by the language $\bw(A) = h(a)$, where $a$ is the element of $\mathbf{K}$ corresponding to formula $A$ (Lindenbaum--Tarski construction). We always have $\overline{a} \in h(a)$ and $h(a)^{\rhd} = \{ (\varepsilon, \underline{a}) \}^{\lhd\rhd} \ni (\varepsilon, \underline{a})$. Therefore, any language of the form $M = \bw(A)$ satisfies the conditions of Lemma~\ref{Lm:grhdlhd}.

Now let us construct a new SCL-model $\Md'$ over the two-letter alphabet $\Sigma_2$, taking $\alpha' = g \circ \alpha$. Let us show that $\bw'$ commutes with the operations.

Obviously, $g$ commutes with meet (set-theoretic intersection): $\bw'(A \wedge B) = g(\bw(A) \cap \bw(B)) = g(\bw(A)) \cap g(\bw(B))$. It also commutes with division, provided its denominator (which is of the form
$\bw(A)$) is non-empty. Indeed, if a word $w$ belongs to $g(M) \SL g(N)$, it should be of the form $g(w')$; otherwise, we append a word from $g(N)$ (which is non-empty) and get a word which is not from the image of $g$, and thus could not belong
to $g(M)$. Having this in mind, we work only in the image of $g$, which is isomorphic to the set of words over the original alphabet.

Finally, as an injective homomorphism, $g$ of course commutes with standard language-theoretic union, product, and Kleene iteration. However, here the closure operator also comes into play:
 \begin{multline*}
  g(\bw(A \vee B)) = g((\bw(A) \cup \bw(B))^{\rhd\lhd}) = ( g(\bw(A) \cup \bw(B)) )^{\rhd\lhd} =\\ (g(\bw(A)) \cup g(\bw(B)))^{\rhd\lhd} = (\bw'(A) \cup \bw'(B))^{\rhd\lhd} = \bw'(A \vee B);
 \end{multline*}
 \begin{multline*}
  g(\bw(A \cdot B)) = g((\bw(A) \cdot \bw(B))^{\rhd\lhd}) = ( g(\bw(A) \cdot \bw(B)) )^{\rhd\lhd} =\\ (g(\bw(A)) \cdot g(\bw(B)))^{\rhd\lhd} = (\bw'(A) \cdot \bw'(B))^{\rhd\lhd} = \bw'(A \cdot B);
 \end{multline*}
 \begin{multline*}
 g(\bw(A^*)) = g((\bw(A)^*)^{\rhd\lhd}) = (g(\bw(A)^*))^{\rhd\lhd} =\\ (g(\bw(A))^*)^{\rhd\lhd} = (\bw'(A)^*)^{\rhd\lhd} = \bw'(A^*).
 \end{multline*}
 The second equality, which is the only interesting one, in each line here is due to Lemma~\ref{Lm:grhdlhd}. Applying Lemma~\ref{Lm:grhdlhd}, however, requires
 checking the non-emptiness conditions for languages $M$ of the form $\bw(A) \cup \bw(B)$, $\bw(A) \cdot \bw(B)$, and $\bw(A)^*$.
 Since $\bw(A)$ and $\bw(B)$ include non-empty words, so does $M$.
  For the non-emptiness of $M^{\rhd}$, the arguments are based on the facts that $M^{\rhd} = M^{\rhd\lhd\rhd}$ and $\bw(C)^{\rhd} \ne \varnothing$ for any formula $C$:
 \begin{align*}
 & (\bw(A) \cup \bw(B))^{\rhd} = (\bw(A) \cup \bw(B))^{\rhd\lhd\rhd} = \bw(A \vee B)^{\rhd} \ne \varnothing; \\
 & (\bw(A) \cdot \bw(B))^{\rhd} = (\bw(A) \cdot \bw(B))^{\rhd\lhd\rhd} = \bw(A \cdot B)^{\rhd} \ne \varnothing; \\
 & (\bw(A)^*)^{\rhd} = (\bw(A)^*)^{\rhd\lhd\rhd} = \bw(A^*)^{\rhd} \ne \varnothing.
 \end{align*}

Finally, the new SCL-model $(\Bc_{g(L_0)}, \alpha')$ validates exactly the same sequents as $(\Bc_{L_0}, \alpha)$ from Theorem~\ref{Th:SCL}. This yields  strong completeness over $\Sigma_2$.
\end{proof}

Next, let us briefly discuss \emph{regular} SCL-models. An SCL-model is called regular, if its language $L$ is a regular language. As shown by Wurm~\cite[Lemma~11]{Wurm2017}, and SCL is regular if{f} it (i.e., its domain) is finite. This is very classical, going back to the famous Myhill -- Nerode theorem.

For finite (regular) SCL-models, we have weak, but not strong completeness. Again, for definiteness we consider the case without Lambek's restriction, as the other one is similar.

\begin{theorem}
 $\ACTomega$ is weakly complete w.r.t.\ finite SCL-models with non-standard $\bot$.
\end{theorem}

\begin{proof}
 We use the FMP for $\ACTomega$~\cite{Buszkowski2007}: any sequent not derivable in $\ACTomega$ is falsified on a finite $\omega$AL. By~\cite[Lemma~12]{Wurm2017} the language $L$ constructed from this $\omega$AL is regular. (We have altered Wurm's construction by adding a second component to $L$, but it is obviously regular.) Any upper cone of a finite SCL is finite. Thus, we obtain the desired finite model which falsifies the sequent.
\end{proof}

Strong completeness fails, however, even for finite sets of hypotheses.

\begin{theorem}
 Neither $\ACTomega$, nor even $\MALC$ is strongly complete w.r.t.\ regular SCL-models.
\end{theorem}

\begin{proof}
 Suppose the contrary. Since regular SCL-models are special kinds of finite ones, this makes $\MALC$ strongly complete w.r.t.\ finite (algebraic) models. Hence, the algorithmic problem of deriving sequents from finite sets of sequents, in $\MALC$, becomes both enumerable (proof search) and co-enumerable (finite countermodel search). By Post's theorem, it is decidable. On the other hand, derivability from finite sets of sequents in $\MALC$ is known to be undecidable~\cite{Buszko82a}. Contradiction.
\end{proof}

A concrete counter-example for strong completeness (though not a short one) can be extracted from the undecidability proof for $\MALC$, which itself goes back to ideas of Markov~\cite{Markov} and Post~\cite{Post}. Namely, there exists a finite set of hypotheses $\Hc$ of the form $p_1, \ldots, p_n \yields q_1 \cdot \ldots \cdot q_m$, which encodes a type-0 grammar for a $\Sigma^0_1$-complete language $S$: a word $r_1 \ldots r_k$ belongs to $S$ if{f} $\Hc \vdash_{\MALC} r_1, \ldots, r_k \to s$, where $s$ is the designated starting symbol. (This construction can be of course adapted to $\MALCp$).

Now we use the fact that any $\Sigma^0_1$-complete set is \emph{creative}~\cite{Myhill}, which means that there exists a computable function $f$ which yields, given the index $i$ of a recursively enumerable set $W_i$, an element $f(i) \notin A \cup W_i$. As $W_i$, we take the set of all sequents of the form $r_1, \ldots, r_k \to s$ such that the claim $\Hc \vDash r_1, \ldots, r_k \to s$ is falsified on a finite model. The sequent $f(i)$ is the needed counter-example. Since $f(i) \notin A$, we have $\Hc \nvdash f(i)$. However, since $f(i) \notin W_i$, this cannot be verified by a finite counter-model.

Let us return to infinite models and extract some corollaries of our completeness results. The first one considers abstract algebraic models of $\ACTomega$ on action lattices which are complete in the lattice-theoretic sense (i.e., which include all infinite suprema and infima). This is a narrower class than the class of $\omega$ALs, since $*$-continuity requires existence only of suprema of a very specific form.\footnote{In the presence of all suprema, the action lattice is always $*$-continuous, even if originally Kleene star was defined in a fixpoint fashion. This is a well-known fact which Restall attributes to Pratt, calling it ``Pratt's normality theorem''~\cite[Thm.~9.44]{Restall}.} Weak completeness w.r.t.\ such models is actually due to Buszkowski~\cite{Buszkowski2007}, who proved the FMP for $\ACTomega$: any finite lattice is complete.
For strong completeness, the FMP does not work. SCLs and their upper cones, however, are complete as lattices, which does the job. Actually, the closure construction used is quite close to Dedekind -- MacNeille completions used in lattice theory.

\begin{cor}
$\ACTomega$  is strongly complete w.r.t.\ models on $\omega$ALs which are complete as lattices; the same holds for $\ACTpomega$ and $\omega$PALs.
\end{cor}

Another corollary is completeness of a fragment of $\ACTomega$ w.r.t.\ L-models (in the original sense).

\begin{cor}
The fragment of $\ACTomega$ in the language of $\BS$, $\SL$, $\wedge$, $\top$, and ${}^*$, where ${}^*$ is allowed to be used only in subformulae of the form
$A^* \BS B$ and $B \SL A^*$, is strongly complete w.r.t.\ L-models.
\end{cor}

\begin{proof}
The idea here is that the listed operations, including the composite operations $A^* \BS B$ and $B \SL A^*$, behave in the same way in L-models and SCL-models. For $\BS$, $\SL$, $\top$, and $\wedge$, this was already mentioned above.

For iteration in denominators, we have to prove the following: $M_1^* \BS M_2 = M_1^{\circledast} \BS M_2$, for any $M_1, M_2 \in \Bc_{L_0}$; the case of $\SL$ is symmetric. The $\supseteq$ inclusion follows from $M_1^* \subseteq M_1^{\circledast}$ and the fact that $\BS$ is antitone on the denominator. Let us prove the $\subseteq$ inclusion. Suppose $u \in M_1^* \BS M_2$ and let $v \in M_1^{\circledast}$. We have to show that $vu \in M_2$. Since $M_2$ is closed, it is sufficient to show $vu \in M_2^{\rhd\lhd}$. Let $(x,y) \subseteq \Sigma^* \times \Sigma^*$ be such a pair that for any $w \in M_2$ we have $xwy \in L_0$.

Let us consider the pair $(x,uy)$. For any $v' = v_1 \ldots v_n \in M_1^*$, since $u \in M_1^* \BS M_2$, we have $w = v'u \in M_2$. Hence, $xwy \in L_0$. Next, we have $v \in M_1^{\circledast} = (M_1^*)^{\rhd\lhd}$, whence $xvuy \in L_0$, which is what we need.
\end{proof}

A similar result holds for the case with Lambek's restriction (it was earlier proved in~\cite{KuznetsovRyzhkova2020}):
\begin{cor}
The fragment of $\ACTpomega$ in the language of $\BS$, $\SL$, $\wedge$, $\top$, and ${}^+$, where ${}^+$ is allowed to be used only in subformulae of the form
$A^+ \BS B$ and $B \SL A^+$, is strongly complete w.r.t.\ L-models.
\end{cor}

Finally, we semantically establish \emph{strong conservativity} of $\ACTomega$ over $\MALC$ and, furthermore, of $\ACTomega$ over any elementary fragment including the product operation.

\begin{cor}
 Consider an elementary fragment $\Lf$ of $\ACTomega$, obtained by restricting the set of operations to a set $\mathrm{Op}$ and keeping only the axioms and rules which work with these operations, plus Id and Cut. Let the product operation belong to $\mathrm{Op}$. Then if $\Hc$ and $\Pi \yields C$ are in this restricted language, then $\Hc \vdash_{\Lf} \Pi \yields C$ if{f} $\Hc \vdash_{\ACTomega} \Pi \yields C$. The same holds for $\ACTpomega$.
\end{cor}

\begin{proof}
 The ``only if'' direction is trivial, let us prove the ``if'' one. Consider the natural class of algebraic models for $\Lf$. Such a model is a partially ordered structure with the operations from $\mathrm{Op}$, which obey the axioms and rules of $\Lf$ (i.e., the rules of $\ACTomega$ working with operations from $\mathrm{Op}$). By Lindenbaum -- Tarski, $\Lf$ is strongly complete w.r.t.\ this class of models.

 Notice that the product here is required, otherwise we, e.g., fail to even formulate the truth condition for a sequent. In the case without Lambek's restriction, we also add a unit for product, even if it is not in the set $\mathrm{Op}$.

 Now we perform our completeness proof (Theorem~\ref{Th:PSCL} or Theorem~\ref{Th:SCL}), for the restricted set of operations. The proof works even simpler, since we have less cases to consider. Again, product and the unit (if there is no Lambek's restriction) are the only two things we really need to be in our abstract model, as they are used to define $w^{\bullet}$. Thus, we construct a homomorphic embedding of our abstract model to the corresponding reduct of a (P)SCL or, if there is $\bot$, of an upper cone of a (P)SCL.

 Given a set of sequents $\Hc$ and a sequent $\Pi \yields C$, such that $\Hc \nvdash_{\Lf} \Pi \yields C$, we obtain an $\mathrm{Op}$-reduct of a (P)SCL-model (with non-standard $\bot$, if $\bot \in \mathrm{Op}$), which validates $\Hc$ and falsifies $\Pi \yields C$. However, in any (P)SCL all the operations of $\ACTomega$ (resp., $\ACTpomega$) are actually well-defined. Thus, this is a model of the full calculus, and therefore we have $\Hc \nvdash \Pi \yields C$ not only in $\Lf$, but also in $\ACTomega$ (resp., $\ACTpomega$).
\end{proof}

This semantic strong conservativity approach can be compared with purely syntactic ones. First, weak conservativity (for derivability without hypothesis) is a standard corollary of cut elimination. For derivability from hypotheses, cut elimination itself does not hold, but strong conservativity can still be proved using a similar technique of \emph{cut normalisation}~\cite{Kuzn2024SEMR}. In~\cite{Kuzn2024SEMR}, however, it was necessary for $\mathrm{Op}$ to contain at least one division operation. Here, instead, we need product.

\section{The $\bot$ Constant}\label{S:zero}

As we have seen so far, Wurm's construction of SCL-models and its modifications (both with and without Lambek's non-emptiness restriction)
fail to give the standard interpretation to the $\bot$ constant, which is the closure of $\varnothing$.
Indeed, this closure gathers the words which belong to $L$ in \emph{any} surrounding context. In~\cite{Kuzn2024WoLLIC}, we have conjectured that adding
$\bot$ to $\MALC$ will make the system incomplete w.r.t.\ SCL-models (with the standard interpretation for $\bot$), at least in the strong sense.
This intuition was also supported by the fact that without $\bot$ the calculus $\MALC$, being an \emph{intuitionistic} variant of non-commutative linear
logic, actually happens to be a conservative fragment of \emph{classical} non-commutative (cyclic) linear logic, while with $\bot$ it is not~\cite{Schellinx1991,KanKuzNigSce2018}.
This makes the behaviour of $\bot$ somehow special.

Interestingly enough, however, in the case without Lambek's restriction one can actually trade $\U$ for $\bot$ and obtain the following strong completeness result. We also have to return to positive iteration instead of Kleene iteration, since the latter contains a hidden $\U$.

\begin{theorem}
The fragment of $\ACTomega$ without $\U$ (but with $\bot$) and with positive iteration instead of Kleene iteration is strongly complete w.r.t.\ SCL-models with the standard interpretation of $\bot$.
\end{theorem}

Thus, $\bot$ itself does not cause problems with completeness. Unfortunately, however, the construction we present below is incompatible both with $\U$ and with
Lambek's non-emptiness restriction; thus, strong completeness for the full systems remains an open question.

The scheme of the proof is again the same: given an abstract algebraic model $\Kc$ over an $\omega$AL, we construct a specific language $L$ and provide an injective homomorphism $h$ from $\Kc$ to
$\BcL$. More precisely, here we are talking about the \emph{reducts} of $\Kc$ and $\BcL$ to the language without $\U$, since $h(\U)$, unfortunately, will not be $\{\varepsilon\}^{\rhd\lhd}$. (One could
also talk of models with non-standard $\U$ and standard $\bot$.)

Again, let $\Sigma = \oSigma \cup \uSigma$.
Let us extend the $^{\bullet}$ notation to arbitrary words over $\Sigma$, consistently with the original notation. We shall simply ignore letters from $\uSigma$. Namely, if $x = w_0 \, \ud_1 \, w_1 \, \ud_2 \ldots w_{n-1} \, \ud_n \, w_n$, where each $w_i$ belongs to $\oSigma^*$, then $x^{\bullet} = (w_0 w_1 \ldots w_{n-1} w_n)^{\bullet}$.

Now the language $L$ is defined as follows:
\[
L = \{ x \ub \mid x \in \Sigma^*, x^{\bullet} \preo b \} \cup \{ z \mid z \in \Sigma^*, z^{\bullet} = \bot \}.
\]
Notice that for words of the form $x \ub$ the second component makes no addition: if $z = x \ub$ and $z^{\bullet} = \bot$, then also $x^{\bullet} = \bot$, and therefore $x^{\bullet} \preo b$.

%w_0 \, \ud_1 \, w_1 \, \ud_2 \ldots w_{n-1} \, \ud_n \, w_n \, \ub \mid n \geq 0; w_0, \ldots, w_n \in \oSigma^*; \\ (w_0 w_1 \ldots w_n)^{\bullet} \preo b \} \cup{} \\
%\{ u \mid u \in \oSigma^*; u^{\bullet} = 0 \} \cup \\
%\{ w_0 \, \ud_1 \, w_1 \ldots \, \ud_n \, w_n \mid w_1, \ldots, w_n \in \oSigma^*; (w_0 w_1 \ldots w_n)^{\bullet} = 0, w_n \ne \varepsilon \}.
%\end{multline*}
The mapping $h$ is defined exactly as before: $h(b) = \{ (\varepsilon, \ub) \}^{\rhd}$; now we have
$h(b) = \{ x \in \Sigma^* \mid x^{\bullet} \preo b \}$. In particular, again always $\oa \in h(a)$.
%\[
% h(b) = \{ w_0 \, \ud_1 \, w_1 \, \ud_2  \ldots %w_{n-1} \, \ud_n \, w_n \mid n \geq 0;  w_0, \ldots, w_n \in \oSigma^*; (w_0 w_1 \ldots w_n)^{\bullet} \preo b \}.
%\]

Now we can see why adding $\U$ ruins this construction: in $h(\U)$, there will be words including letters from $\uSigma$. If we require $h(\U) = \{ \varepsilon \}^{\rhd\lhd}$, then such a word could be appended to the end of a word in $L$. This is undesired, as the \emph{last} letter $\ub$ from $\uSigma$ plays a very special r\^{o}le, while other such letters (those inside $x$) are ignored. Changing the last
letter from $\uSigma$ will most probably violate the $x^{\bullet} \preo b$ condition.

Imposing Lambek's non-emptiness restriction also jeopardises our construction. Namely,
$h(b \SL b) = h(b) \SL h(b)$ would include words which contains no letters from $\oSigma$, but some letters from $\uSigma$. Such a word itself is non-empty, so it is a legal word in the model, but $w_0 \ldots w_n$ is indeed empty, and
without the unit $(w_0 \ldots w_n)^{\bullet}$ is meaningless. Without $\U$ and without Lambek's restriction, however, the construction works, and gives the standard interpretation for $\bot$.

%For brevity, if a word $x$ includes symbols from both $\oSigma$ and $\uSigma$, i.e., $x = w_0 \, \ud_1 \, w_1 \, \ldots \, w_{n-1} \, \ud_n \, w_n$, where $w_i \in \oSigma^*$, let
%$x^{\bullet} = (w_0 w_1 \ldots w_n)^{\bullet}$ (this notation is consistent with the definition of $w^{\bullet}$ for $w \in \oSigma^*$). In this notation,
%\( h(b) = \{ x \in \Sigma^* \mid x^{\bullet} \preo b \).

\begin{lemma}\label{Lm:replzero}
Let $x \in \Sigma^*$ and $u \in \oSigma^+$. Then $x^{\bullet} \preo u^{\bullet}$ in $\Kc$ if{f} $x \leqslant_L u$.
\end{lemma}

\begin{proof}
For the `only if' direction, let $sut \in L$ for some $s,t \in \Sigma^*$.
If $(sut)^{\bullet} = \bot$, then by monotonicity and transitivity of product we also have $(sxt)^{\bullet} \preo \bot$, i.e., $(sxt)^{\bullet}  = \bot$, whence $sxt \in L$.
In the other case, $sut$ should be of the form $y\ub$, where $y^{\bullet} \preo b$. Thus, $t = t'\ub$ and $(sut')^{\bullet} \preo b$. (Notice that here we use the fact that $u$ is over $\oSigma$ and it is non-empty: otherwise $t$ could have been empty, and the rightmost $\ub$ could belong to $u$ or $s$.) Again by monotonicity and transitivity, we get $(sxt')^{\bullet} \preo b$, whence $sxt = sxt'\ub \in L$.

For the `if' direction, let $b = u^{\bullet}$. We have $u \ub \in L$ and, since $x \leqslant_L u$, also $x\ub \in L$. Hence, $x^{\bullet} \preo b = u^{\bullet}$.
\end{proof}

Now we prove the key statement about $\bot$:
\begin{lemma}
 $h(\bot) = \varnothing^{\rhd\lhd}$.
\end{lemma}

\begin{proof}
 Take an arbitrary pair $s,t \in \Sigma^*$ and let $z \in h(\bot)$. The latter means that $z^{\bullet} = \bot$. Hence, $(szt)^{\bullet} = \bot$ ($\bot$ is the zero for product), therefore $szt \in L$.
\end{proof}

Finally, we routinely check that everything else (except $\U$ and Kleene iteration, which is replaced by its positive version) keeps working.

\begin{lemma}
 The mapping $h$ preserves all operations of $\Kc$, and if $a \preo b$ in $\Kc$, then $h(a) \subseteq h(b)$.
\end{lemma}

\begin{proof}
 The proof is even simpler than that of Lemma~\ref{Lm:hom}, because now we do not need to consider two cases.

 If $a \preo b$ in $\Kc$ and $x \in h(a)$, then $x^{\bullet} \preo a \preo b$, whence $x \in h(b)$. For operations, the argument is as follows.

 \vskip 5pt
 \emph{Product.} If $x \in h(a \cdot b)$, then we have $x^{\bullet} \preo a \cdot b$, and by Lemma~\ref{Lm:replzero} $x \leqslant_L \oa\ob$. Having $\oa\ob \in h(a) \cdot h(b)$, we conclude that $x \in h(a) \circ h(b)$. Now take $x = x_1 x_2$, where $x_1 \in h(a)$, $x_2 \in h(b)$. We have $x_1^{\bullet} \preo a$, $x_2^{\bullet} \preo b$, whence $x^{\bullet} = x_1^{\bullet} x_2^{\bullet} \preo a \cdot b$. Therefore, $x \in h(a \cdot b)$.

 \vskip 5pt
 \emph{Division.} Take any $x \in h(b \SL a)$ and any $y \in h(a)$. We have $x^{\bullet} \preo b \SL a$, $y^{\bullet} \preo a$. Hence, $(xy)^{\bullet} \preo (b \SL a) \cdot a \preo a$. This shows that $h(b \SL a) \subseteq h(b) \SL h(a)$. For the opposite direction, take $x \in h(b) \SL h(a)$ and $y = \oa \in h(a)$. We have $x \oa \in h(b)$, whence $x^{\bullet} \cdot a \preo b$, which is equivalent to $x^{\bullet} \preo b \SL a$. Therefore, $x \in h(b \SL a)$. The $\BS$ case is symmetric, since so is the expression for $h(b)$.

 \vskip 5pt
 \emph{Intersection:} the same as in the proof of Lemma~\ref{Lm:hom}, with $w \in \oSigma^*$ changed to an arbitrary $x \in \Sigma^*$.

 \vskip 5pt
 \emph{Union.}
 Let us show that $\overline{a \vee b}$ belongs to $h(a) \sqcup h(b)$. Take a pair $s,t$ such that for any $x \in h(a) \cup h(b)$ we have $sxt \in L$. If $t = t' \ud$, then we have $(s\oa t')^{\bullet} \preo d$ and $(s \ob t')^{\bullet} \preo d$. Hence, $(s\, \overline{a \vee b} \, t')^{\bullet} \preo d$, and finally $s \, \overline{a \vee b} \, t \in L$. In $t$ is empty or ends on a letter from $\oSigma$, we must have $(s \oa t)^{\bullet} = (s \ob t)^{\bullet} = \bot$, whence $(s \, \overline{a \vee b} \, t)^{\bullet} = \bot$, and we again have $s \, \overline{a \vee b} \, t \in L$. The rest of the argument resembles that from Lemma~\ref{Lm:hom}.

 %For $h(a \vee b) \subseteq h(a) \sqcup h(b)$, the argument is the same as in the proof of Lemma~\ref{Lm:hom}. For the opposite inclusion, take $x \in h(a)$ (or, symmetrically, $x \in h(b)$). We have $x^{\bullet} \preo a \preo a \vee b$, whence $x \in h(a \vee b)$.
 %FIXME make argument

 \vskip 5pt
 \emph{Positive iteration.}
 The argument is similar to that for union. We show that $\overline{a^+}$ belongs to $(h(a))^{\oplus}$, and use the fact that if $x \in h(a^+)$, then $x \leqslant_L \overline{a^+}$. Notice that for Kleene iteration this would not work: if $s,t$ are such that for any $n \geq 0$ and for any $x \in h(a)^n$ we have $sxt \in L$, then we could have an empty $t$, and for the empty $x$ we could have $s^{\bullet} \ne \bot$ ($s$ ends with a letter from $\uSigma$).
 %
 %
%
 %FIXME change
 %Here we show that $\overline{a^*}$ belongs to $(h(a))^{\circledast}$. Take a pair $s,t$ such that for any $n \geq 0$ and for any $x \in h(a)^n$ we have $sxt \in L$. If $t = t'\ud$, then we have $(s \oa^n t')^{\bullet} \preo d$ for any $n \geq 0$. By the $\omega$-rule, $(s \overline{a^*} t')^{\bullet} \preo d$, whence $s \overline{a}^* t \in L$. If $t$ is empty,
 %The same modification of the argument for the union, as in the proof
 %FIXME: do accurately! could be problem with unit

\vskip 5pt
 Notice that now we do not need to consider the case of $\top$ separately, since $\top$ is equivalent to $\bot \SL \bot$, and both $\bot$ and $\SL$ map correctly.
\end{proof}

The rest is done in the same way as in the proof of Theorem~\ref{Th:PSCL}. A tricky moment happens, however, with sequents of the form ${} \yields B$. Such a sequent is declared true in the SCL-model, if $\varepsilon \in \alpha(B)$ (or, which is the same, $\{\varepsilon\}^{\rhd\lhd} \subseteq \alpha(B)$, see Section~\ref{S:norestr}). Our construction, however, yields something which looks different.
Despite the absence of the constant $\U$ in our language, the $\omega$AL $\Kc$ still has the unit element $\U$, and the sequent ${} \yields B$ is interpreted in $\Kc$ as $\U \preo \beta(B)$. In our SCL-model, this translates to $h(\U) \subseteq \alpha(B)$, and, unfortunately,
$h(\U) \ne \{ \varepsilon \}^{\rhd\lhd}$.

We shall prove, however, that $h(\U) \subseteq \alpha(B)$ if{f} $\varepsilon \in \alpha(B)$ (where $\U$ is the unit of $\Kc$). Since $\varepsilon^{\bullet} = \U$, we have $\varepsilon \in h(\U)$, which validates the ``only if'' direction. Let us prove the ``if'' direction. Let $b = \beta(B)$, then $\alpha(B) = h(b)$ and $\varepsilon \in h(b)$. The latter means that $\U = \varepsilon^{\bullet} \preo b$. Now, since for any $x \in h(\U)$ we have $x^{\bullet} \preo \U$, by transitivity we get $x^{\bullet} \preo b$, whence $x \in h(b) = \alpha(B)$.

\subsection*{Acknowledgments}
The author thanks C.~Wurm, S.~O.~Kuznetsov, and the participants of WoLLIC 2024 for fruitful discussions, and the reviewers for their valuable comments.

Support from the Basic Research Program of HSE University is gratefully acknowledged.

\bibliographystyle{abbrv}
\bibliography{SCLmodels}

\end{document}